\documentclass{sig-alternate}
\usepackage{graphicx}

\newtheorem{theorem}{Thereom}

\newtheorem{lemma}{Lemma}

\newtheorem{claim}{Claim}
\newdef{definition}{Definition} 
\newcommand{\E}{\mathbb{E}} 
\newcommand{\R}{\mathbb{R}}
\newcommand{\BF}[2]{\mbox{\textbf{BF}}(#1,#2)}
\begin{document}

\title{Accelerating Nearest Neighbor Search on Manycore Systems}

\numberofauthors{1}
\author{
\alignauthor
Lawrence Cayton\\
\affaddr{Max Planck Institute}\\
\affaddr{T\"ubingen, Germany}\\
\email{lcayton@tuebingen.mpg.de}
}

\maketitle
\begin{abstract}


We develop methods for accelerating metric similarity search that are 
effective on modern hardware.  Our algorithms factor into easily
parallelizable components, making them simple to deploy and efficient 
on multicore CPUs and GPUs.  Despite the simple structure of our algorithms,
their search performance is provably sublinear in the size of the
database, with a factor dependent only on its \emph{intrinsic} dimensionality.  
We
demonstrate that our methods provide substantial speedups on a range
of datasets and hardware platforms.  In particular, we present
results on a 48-core server machine, on graphics hardware, and on a
multicore desktop.  
\end{abstract}

\section{Introduction}
We study methods to accelerate nearest neighbor search on modern
parallel systems.  We work in the standard metric nearest neighbor (NN)
setting: given a database $X$, return the closest point to any query
$q$, where closeness is measured with some fixed metric.  Though the
problem setting is by now well-studied, there are two recent developments
that provide a different focus for this work.  The first development
comes from studies in data analysis and advances in theoretical
computer science, where the notion of intrinsic dimensionality has
been refined and profitably exploited to manage high-dimensional data.
The second development comes from the computer hardware industry: CPUs
are virtually all multicore now and GPUs have rapidly evolved into
powerful secondary devices for numerical computation, forcing a
renewed interest in parallelism.  

High-dimensional data challenges nearly all methods for
accelerating NN search.  Unfortunately, complicated, high-dimensional
data is the norm in many domains that rely on NN 
search, such as computer vision \cite{trf.2008}, bioinformatics
\cite{b.2001}, and data analysis in general \cite{d.00}.  Because
such data must be dealt with, researchers have explored data
properties that might be exploited to accelerate search.  A
compelling property is the intrinsic dimensionality 
of a data set; the idea is that often data only \emph{appears}
high-dimensional (i.e., each element has many features), but is
actually governed by only a few parameters.  In recent years, this
type of data has been studied extensively and is now believed to be
widespread; in machine learning, for example, several methods have
been developed to reveal the intrinsic structure of data
\cite{rs.00,tdsl.00}.  In research on accelerated NN retrieval, a
renewed focus on intrinsic dimensionality in recent years has yielded
methods with strong theoretical guarantees \cite{KR02,KL04} and state
of the art empirical performance \cite{BKL06}.  

The end goal of all of this work on NN is, of course, to make NN
search run fast.  While properties of data are important for
computational efficiency, equally important is the machine hardware on
which search actually runs.  Hardware properties are especially
important now, as the machines in use for everyday data analysis and
database operations are fundamentally different than they were 
even a decade ago.  In particular, standard computer chips are
multicore with vector (SIMD) units, and GPUs have become popular as 
secondary devices for data-intensive processing.  These modern
processors have very impressive computational capabilities, but
fully exploiting it requires considerable parallelism.  


The shift towards parallelism in hardware necessitates a refocusing in
algorithm design and software engineering methods \cite{p.10}.  At the
level of algorithm design, advances in algorithms and data structures
may not be useful unless they can be effectively parallelized.  At the
level of software development, developing efficient implementations on 
parallel systems is challenging, so parallel primitives and design
patterns are required to ease the burden on programmers.  For the case
of NN search, the modern search algorithms discussed were developed
for sequential systems, and seem quite difficult to deploy on modern
hardware.  This is a major practical limitation.  


In this work, we develop an new approach to NN search that is both
provably sensitive to intrinsic structure \emph{and} that is effective
on modern CPUs and GPUs.  Our algorithms are based on fundamental
ideas from metric similarity search, and are designed carefully for
two important goals.  First, our design choices allow us to rigorously bound
the search complexity  of our methods.  Equally important, these
choices allow us to factorize our algorithms into a basic primitive
that is simple to parallelize, making them effective and relatively
easy to implement on different systems.  

Our methods are of immediate practical use.  We demonstrate their
performance benefits on a range of modern platforms: a 48-core server
machine, a GPU, and a multicore desktop.  

\section{Background and Related Work}



The focus of the present work is unique; as far as we know, it is the
only work to simultaneously make use of modern algorithmic
developments and modern hardware.  Still, the ideas 
behind our methods are based on a substantial amount of previous work:
the algorithmic ideas are based on techniques from metric similarity
search; the results on intrinsic dimensionality are related to 
ideas developed mostly in the theory and data analysis communities;
and the motivation behind the project comes from recent trends in
computer hardware, especially as related to databases and scientific
computation.     

The data structure and algorithm in this work are based on two
fundamentals of algorithms for metric data: space decomposition and
the triangle inequality.  These pillars are used in virtually all work
on metric NN search; see the surveys of Ch\'{a}vez et al.\ and Clarkson
for detailed overviews \cite{cnbm.01,C06}.  Two of the most
empirically effective structures are AESA \cite{V86} and metric ball
trees \cite{O89, Y93}, both of which have spawned many relatives.  

A long-standing problem in similarity search is the difficulty of 
dealing with high-dimensional data; see \cite{AMNSW98,b.95,wsb.98} and
the above surveys.  The basic challenge is that space-decomposition
structures that reduce the work for NN retrieval seem to have
performance that scales exponentially with the
dimensionality of the data, rendering them useless to all but the
smallest problems.  Within the last two decades there have been two
very promising directions of work that attempt to deal with the
problem of high-dimensional data.   

The first is called Locality Sensitive Hashing (LSH) \cite{indyk98}.
LSH has retrieval performance that is provably sublinear,
independent of the underlying dimensionality.  This was a major
theoretical breakthrough and the data structure has been successfully
deployed on some tasks (e.g. \cite{b.2001}).  However, LSH has some
limitations: it can only provide approximate answers, it is defined
only for particular distance measures (not at the generality of
metrics), and setting the parameters correctly can be complex \cite{d.2008}.

The second line of work, upon which we build, is based on the
notion of \emph{intrinsic dimensionality}.  The basic idea here is
that many data sets only appear high-dimensional, but are actually
governed by a small number of parameters.  Within data
analysis and machine learning, the idea of low-dimensional intrinsic
structure has become extremely popular and such structure is believed
to be common in many data sets of interests \cite{rs.00,tdsl.00}.

This idea has also been explored in the context of NN search.  A
variety of slightly different notions of metric space dimensionality
capture this intuition.  One which has recently resulted in strong
theoretical and empirical results is the \emph{growth dimension} 
or \emph{expansion rate}  \cite{c.99,KR02}, which we define formally later.
This notion of dimension lead to the development of the Cover Tree
\cite{BKL06,KL04}, which we return to momentarily.  Though the
notion has some idiosyncrasies \cite{KL04}, the impressive empirical
performance of the Cover Tree suggests that it is a useful notion.  

Perhaps the two most relevant methods for NN search are the Cover Tree
and the GNAT of Brin \cite{b.95}; let us distinguish this research from the
present work.  The GNAT uses a simple space decomposition based on
representatives from the database, much as we do, and also discusses
the idea of intrinsic dimensionality.  However, the relationship of
the GNAT's search performance and the intrinsic dimensionality is only
discussed in an informal, heuristic way, whereas we give rigorous
runtime guarantees.  These rigorous bounds require a search algorithm
that is different than that of the GNAT.  Additionally,
parallelization is not discussed in \cite{b.95}.

The Cover Tree has rigorous guarantees on the query time dependent on
the expansion rate and has empirical performance that is state of the
art.  Even so, our algorithms, data structure, and theory are substantially
different; in particular, the Cover Tree is a deep tree and is explored
in a conditional way that seems quite difficult to parallelize.  We
compare our method with the Cover Tree in the experiments.

Lastly, we touch on the major inspiration for this paper: the use
of hardware to accelerate data-intensive processes.  Impelled by the
sudden ubiquity of multicore CPUs and the development of GPUs for
general-purpose computation, this area of research has exploded in the
last decade; let us provide a couple of inspiring examples.  A relatively early
work develops methods to off-load expensive database operations onto
the GPU \cite{glwlm.04}.  A very recent piece of work tunes  
basic tree search algorithms (such as for index lookup) to be
effective on modern multicore CPUs and GPUs \cite{kcss.10}.  Finally,
another paper suggests simply running brute force search on a GPU to
accelerate NN search \cite{bdhk.06}; this simple approach provides a
surprising amount of acceleration over computation on sequential CPUs 
\cite{gdb.08}.  


\section{The brute force primitive}
Brute force search requires a high amount of work, but parallelizes
effectively.  In contrast, most accelerated approaches to NN search
reduce the work, but seem difficult to parallelize.  Our approach to
NN search takes the advantages of both: it reduces the work \emph{and}
parallelizes effectively.  We achieve this combination by structuring
our algorithms as multiple (actually two) brute force searches, each of which
considers only a small portion of the database.  In this section, we
formalize the \emph{brute force primitive} and discuss the
parallelization of it. 


Given a set of queries $Q$ and a database $X$ with $n$ elements,
finding the NNs for all $q\in Q$ can be achieved by a series of linear
scans.  For each query $q$, the distance between $q$ and each $x\in X$
is computed, the distances are compared, and the database point that 
is closest is returned.  We denote this subroutine as $\BF{Q}{X}$.  If
$L$ is some sets of IDs (i.e. $L\subset \{1,\ldots,n\}$), then brute
force search between $Q$ and this subset of the database is denoted
$\BF{Q}{X[L]}$.  

The work required for $\BF{Q}{X}$ is $O(n)$; we later prove that our
search algorithms have work only roughly $O(\sqrt{n})$, which is
performed in two brute force calls $\BF{Q}{X[L_1]}$ and
$\BF{Q}{X[L_2]}$, where lists $L_1$ and $L_2$ are determined by our
algorithm.  

Parallelizing $\BF{Q}{X}$ is straightforward.  We break the procedure
down into two steps: a distance computation step, and a comparison
step.  In the distance computation step, all pairs of distances are
computed.  This has virtually the same structure as matrix-matrix
multiply, and hence block decomposition approaches are effective.  In
the case where there is only a single query presented at a time
(e.g. a stream of queries), the distance computation step of
$\BF{q}{X}$ has the structure of a matrix-vector multiplication.  In
both cases, the parallelization is extremely well-studied.  

The second step is the comparison: for each query, the distances must
be compared, and the nearest database element returned.  This is
simple to do in parallel systems as well; the problem can simply be
plugged into the standard parallel-reduce paradigm where comparisons
are made according to an inverted binary tree.  

The computational structure of accelerated NN data structures is quite
different.  For simplicity, let us take metric trees as an example
\cite{O89}.  Querying this data structure requires a depth-first
search of a deep tree.  This process involves an interleaved series
of distance computations, bound computations, and distance
comparisons.  Moreover, the computation is conditional: the specific 
calculations made at one step depend on the comparisons from the
previous step.  Distributing the work for this process across
processors is a significant challenge.  Additionally, because of the
interleaved nature of the computation, attaining near full-utilization
of the hardware is difficult; in contrast, matrix-matrix multiply is
commonly used to demonstrate the capability of processors.
Finally, the conditional nature of the computation makes execution on
vector hardware with limited branching abilities---namely,
GPUs---inefficient.   

Efficient NN routines seem to depend on a complex computational
structure; a major contribution of the present work is in
demonstrating that a much simpler structure can be substituted without
significant loss.  

With the brute force primitive in place, we proceed to discuss our
data structure and algorithms, all of which will be built from this
primitive.  
\section{Data structure}
We discuss the data structure underlying our methods in this section,
which we call the \emph{Random Ball Cover} (RBC).  This is a very
simple, single-level cover of the underlying metric space.  The basic
idea is to use a random subset of the database as representatives for
the rest of the DB.  The details of this structure differ slightly for
our two different search algorithms, which will be introduced in the next
section.  

The database is denoted $X=\{x_1,\ldots, x_n\}$ and the metric in use
is denoted $\rho(\cdot, \cdot)$.  The data structure consists of a random
subset of the database, which will generally be of size about
$O(\sqrt{n})$;  we make this precise later.  This set of random
representatives will be denoted $R$.  It is built by choosing each
element of the database independently at random with probability
$n_r/n$, where the exact value of $n_r$ is discussed in the theory
section.  In expectation, there will be $n_r$ representatives
chosen---one can think of the symbol $n_r$ as shorthand for
\textbf{n}umber of \textbf{r}epresentatives.   

Each representative \emph{owns} some subset of the database.  The
list of points that a representative $r$ owns is denoted $L_r$.  In
the exact search algorithm, $L_r$ contains all $x\in X$ for which $r$
is the nearest 
neighbor among $R$.  In the one-shot algorithm, $L_r$ contains a
suitably sized set of $r$'s NNs among $X$.  We will at times refer to
$L_r$ as an \emph{ownership list}.  

Along with each representative, a radius is also stored.  This radius
is defined as the distance from $r$ to the furthest point that it
owns:
\begin{align*}
\psi_r = \max_{x\in L_r} \rho(x,r).
\end{align*}

The focus of this paper is on algorithms that are simple to
parallelize; this is achieved by folding the computational work into
brute force calls.  The building routines demonstrate this principle
concisely.  The building routine for the exact search algorithm must
find the NN for each $x\in X$ among the representatives $R$.  Thus
this routine is simply a call to $\BF{X}{R}$.  Similarly, the building
routine for the one-shot algorithm must find the NN(s) for each
representative $R$ among the database elements $X$; thus this
procedure is simply a call to $\BF{R}{X}$.  Both parallelize easily.

With the data structure and notation in place, we proceed to describe
the search algorithms.  
\section{Search Algorithms} \label{sec:algs}
In this section, we describe two search algorithms which use the RBC
data structure.  The theoretical analysis of 
these algorithms appears in the next section.  Both of these
algorithms build up from the brute force search primitive.

We first describe the \emph{one-shot} algorithm, then the \emph{exact}
search algorithm.  The one-shot algorithm is 
extraordinarily simple, the exact algorithm only slightly less so.
Both algorithms rely on a randomized data structure, and so are
probabilistic.  In the one-shot algorithm, the solution itself is
randomized: the data structure returns a correct result with high
probability.  In the exact algorithm, the solution is guaranteed to be
correct; only the running time is probabilistic.\footnote{This
  algorithm can be easily modified so that it only guarantees an
  \emph{approximate} nearest neighbor, which reduces search time.}
Hence when an exact answer is required, the exact algorithm is
appropriate; if a small amount of error can be tolerated, the one-shot
algorithm is simpler and often faster, as we show in the experiments.  

We focus on the problem of 1-NN search throughout; the extensions to
$k$-NN and $\epsilon$-range search are straightforward.  

\subsection{One-shot search}
First, recall the data structure built for the one-shot search
algorithm.  The representatives are chosen at random, and each list
$L_r$ contains the $s$ closest database elements to $r$.  Depending on
the setting of $s$, points will typically belong to more than one
representative.  We discuss these parameters further in the theory
section.  

On a query $q$, the algorithm proceeds as follows.  It first computes
the NN to $q$ among the representatives using a simple linear scan
(brute force search), call it $r$.  It then scans the ownership list
$L_r$, computing the distance from $q$ to each listed database point.
The nearest one is returned as the nearest neighbor.  

We restate the algorithm in terms of the brute force primitive.  The
algorithm first calls $\BF{q}{R}$, which returns a 
representative $r$.  It then calls $\BF{q}{X[L_r]}$ and returns the answer.  

The one-shot algorithm is almost absurdly simple.  Yet, as we show
rigorously in the theory section, it provides a massive speedup; in
particular, it reduces the work from $O(n)$ (required for a full brute
force search) to roughly $O(\sqrt{n})$.  Moreover, it is very fast
empirically, as we show in the experiments section.  

\subsection{Exact search}
Whereas the one-shot algorithm does not use the triangle inequality
(though the analysis requires it), the exact search algorithm
explicitly prunes portions of the space using it; in this sense, it is
reminiscent of classic branch-and-bound techniques.  

Recall that the data structure built for the exact search algorithm is 
slightly different from the one for the one search algorithm (the
reason for the difference will become clear in theory section).  The 
build algorithm calls $\BF{X}{R}$, then adds each $x\in X$ to the
ownership list of its returned NN in $R$.  The radii $\psi_r$ are set
to $\max_{x\in L_r} \rho(x,r)$ as before.  

We now detail the search algorithm.  
First, the closest point to $q$ among all $r\in R$ is computed; call
it $r_q$, and let $\gamma = \rho(q,r_q)$.  This distance is an upper
bound on the distance to $q$'s NN (since $r_q \in X$), and so the
algorithm can use it to discard some of the database from consideration.
Recall that the radius of each representative $r$ is stored as
$\psi_r$---\textit{i.e.} each $x \in L_r$ satisfies $\rho(x,r) \leq
\psi_r$.  Because $\gamma$ is an upper bound on the distance to the
NN, any point belonging to an $r$ satisfying
\begin{align}
  \rho(q,r) \geq \gamma + \psi_r \label{ineq:prune}
\end{align}
can be discarded.  The following sketch illustrates
(\ref{ineq:prune}); since there is a point within distance $\gamma$ of
$q$, no points within distance $\psi_r$ of $r$ can be $q$'s NN. 
\begin{center}
\includegraphics[width=.45\linewidth,height=!]{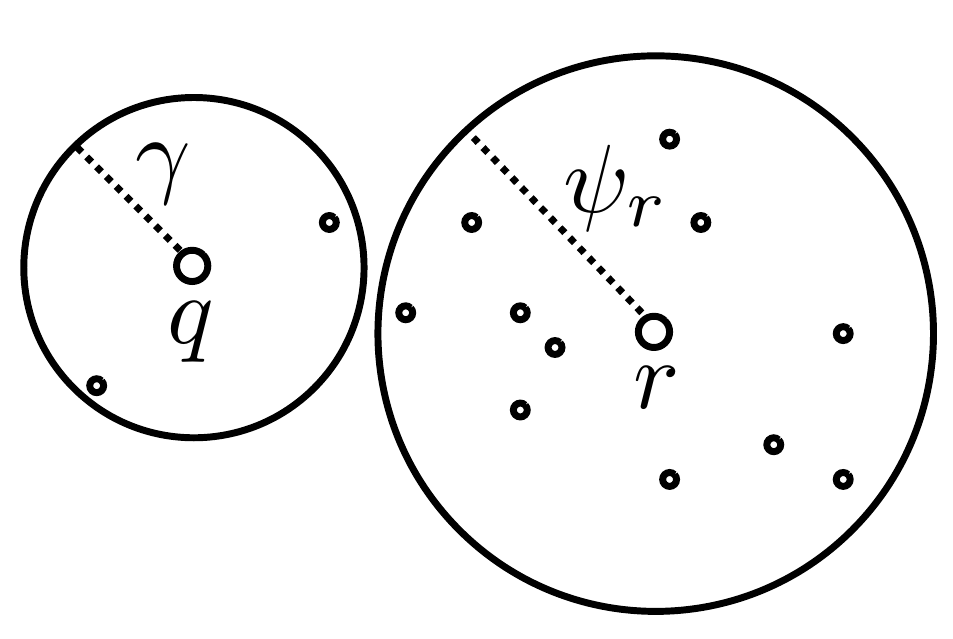}
\end{center}
Hence the only points that need to be considered belong to a list
$L_r$, where $r$ violates (\ref{ineq:prune}).  

The algorithm simultaneously checks a second bound in hopes of pruning
out more of the representatives.  As Lemma \ref{lem:prune2} (see the
next section) shows, if $r_q$ is $q$'s NN among the 
representatives, any representative that owns $q$'s NN must satisfy 
\begin{align}
\rho(q,r) \leq 3\cdot \rho(q,r_q). \label{ineq:prune2}
\end{align}
Hence any representative violating this inequality is pruned out by
the search algorithm.  

Once the pruning stage is complete, the search algorithm computes the
distance to all points belonging to one of the remaining
representatives, and returns the nearest.

We restate the algorithm in terms of our primitive.  It first computes
$\BF{q}{R}$, much like the one-shot algorithm.  In this case, however,
the distances must be retained so that inequalities
(\ref{ineq:prune}) and (\ref{ineq:prune2}) can be checked.  Once
the inequalities are checked, some representatives will still remain,
with lists $L_1, \ldots, L_t$.  Next the search algorithm performs another
brute force search, namely $\BF{q}{X[L_1\cup L_2 \cup \ldots \cup
    L_t]}$, and returns the answer.  

We will see that the size of each of the brute force calls is 
about $O(\sqrt{n})$, providing major time savings over a full brute
force search.  

We emphasize that the computation structure of both search algorithms
is quite different from tree-based search, in which bounds are
incrementally refined, and distance computations are interleaved with
bound evaluations.  In both cases, this structure makes the algorithm
extremely simple to implement and effective to parallelize.  It is
rather surprising that such simple algorithms can effectively reduce
the work required for NN search, but that is exactly what we show both 
theoretically and empirically in the following sections.

\section{Theory} \label{sec:theory}
As we described in the background section, all methods that reduce the
work for NN search have some dependence on the dimensionality of the
database.  Much of the success of metric indexing methods is commonly
ascribed to their dependence  only on the intrinsic dimensionality of
data.  In this section, we prove that the RBC search algorithms scale
with the \emph{expansion rate}, which is a useful notion of intrinsic
dimensionality.   

\begin{definition} Let $B(x,r)$ denote the closed ball of radius $r$
  around $x$---i.e. the set $\{y\;:\; \rho(x,y)\leq r\}$.  A
  finite metric space $M$ has expansion rate $c$ if for all $r>0$ and $x\in M$
\begin{align*}
|B(x,2r)| \leq c \cdot |B(x,r)|.
\end{align*}
\end{definition}
To gain some intuition for this measure, consider a grid of points in
$\R^d$ under the $\ell_1$ metric
\begin{align*}
\rho(x,y) = \sum_{i=1}^d |x_i-y_i|.
\end{align*}
The expansion rate in this case is $2^d$, hence $\log c$ corresponds
to the dimensionality of the data \cite{KR02}.  Notice that the
expansion rate is defined only in terms of the metric, not in terms of
the representation of data; in this sense, the rate captures the intrinsic
structure of the metric space. Notice also that the expansion rate is
defined for arbitrary metric spaces, so makes sense for the edit
distance on strings and the shortest path distance on the nodes of a
graph, for example.   

Throughout we assume that $X \cup Q$ has expansion rate $c$, and we
prove bounds dependent on this expansion rate and $n$.  

The exact search algorithm and analysis rely on the following lemma,
which is known \cite{C06}.  See Appendix \ref{ap:lem-proof} for the proof.
\begin{lemma}
Let $R \subset X$ and assign each $x\in X$ to its nearest $r\in R$.
Let $\gamma = \min_{r\in R} \rho(q,r)$ (i.e., $\gamma$ is the distance
to $q$'s NN in $R$).  Then, if some $r^*\in R$ owns the nearest neighbor
to $q$ in $X$, it must satisfy
\begin{align*}
\rho(q,r^*) \leq 3\gamma.
\end{align*} \label{lem:prune2}
\end{lemma}

We now analyze the search algorithms.  Throughout, we assume $R$ is a
random subset of $X$, built by picking each element of $X$
independently at random with probability $n_r/n$.  Sometimes we will
denote this probability as $p$.  Recall that the ownership list of
$r\in R$ is denoted $L_r$ and the radius of this list
(i.e. $\max_{x\in L_r} \rho(x,r)$) is denoted $\psi_r$.  Finally,
$n_r$ is the expected number of representatives and $n$ is the
cardinality of the database.
 
\subsection{Exact Search}
First, let us consider the exact search algorithm.  The search
algorithm performs two steps: in the first step, the algorithm
performs brute force search from the queries to $R$; and in the second
step, it performs a brute force search from the queries to
the database elements belonging to ownership lists of un-pruned
representatives.  The first step clearly has work complexity $O(n_r)$
per query, where $n_r$ is the expected number of representatives; the
following analysis bounds the complexity of the second step.  In
particular, we show  that the expected number of distance evaluations
is $c^3 n/n_r$. 
Hence if $n_r \approx c^{3/2}\sqrt{n}$, the expected number of distance
evaluations in the second step is $O(c^{3/2} \sqrt{n})$, the same as the
first step.  We call $n_r = O(c^{3/2}\sqrt{n})$ the \emph{standard
  parameter setting}.  


In the first step of the algorithm, the nearest point to $q$ in $R$ is
found; call this point $r_q$.  How many database points are likely to be
closer to $q$ than $r_q$?  

\begin{claim} Let $\gamma$ be the distance from $q$ to its nearest
  neighbor in $R$, $r_q$.  
  The expected number of points in $B(q,\gamma)$ is $n/n_r$, which
  is  $O(\sqrt{n}/c^{3/2})$ for the standard parameter setting.  \label{cl:numPts}
\end{claim} 
\begin{proof}
Form a list $L = [x_1,x_2, \ldots,x_n]$ by ordering the database points
$x\in X$ by their distance to $q$.  Some subset of $X$ also belongs to
$R$; let $x_t$ be the first representative appearing in $L$ (i.e. the
closest representative to $q$).  Then the expected number of points in
$B(q,\gamma)$ is equal to $t-1$.

A slightly different way to view the process is that the $L$ is fixed,
then $x_1$ is chosen as a representative with probability $n_r/n$,
then $x_2$ is chosen as a representative with probability $n_r/n$, and
so on.  We wish to know the expected time before the first $x_i$ is
chosen as a representative.  That is given precisely by the geometric
distribution: the number of Bernoulli trials needed to get one
success.  The mean of a Bernoulli distribution with parameter $p =
n_r/n$ is $1/p = n/n_r$.  Hence $\E|B(q,\gamma)| = n/n_r$, which is 
$O(\sqrt{n}/c^{3/2})$ for the standard parameter setting.  
\end{proof}
We note that a high-probability version of the above claim follows
easily from standard concentration bounds.  We also point
out that the expectation is over randomness in the algorithm; we are
not making any distributional assumptions on the database.

After computing the nearest neighbor to the query $q$ among the
representatives, the exact search algorithm uses $\gamma$ $(\equiv
\rho(q,r_q))$ as an upper bound on the distance to $q$'s NN to prune
out some representative sets.  In particular, any representative $r$ with
radius $\psi_r$  satisfying
\begin{align}
\rho(q,r) \geq \gamma + \psi_r \label{ineq:pruneTheory}
\end{align} 
cannot possibly own $q$'s NN.  Additionally, the algorithm can safely
prune out any representative $r$ such that 
\begin{align}
\rho(q,r) > 3 \gamma \label{ineq:pruneTheory2}.
\end{align}
This property follows from Lemma \ref{lem:prune2}.  In the following
we only work with inequality (\ref{ineq:pruneTheory2}).  The simultaneous
use of both inequalities improved the empirical performance, but it is
not necessary for the following theory. 

We now estimate how many representatives \emph{violate}
(\ref{ineq:pruneTheory2}) and bound how many points these
representatives own.  We show that all examined database points belong to a
ball $B(q,7\gamma)$, which we subsequently bound the cardinality of.  

\begin{claim}
The nearest neighbor of $q$ lies inside of the ball
$B(q,7\gamma)$. \label{cl:ballBnd}  
\end{claim}
\begin{proof}
Clearly, each representative $r$ violating (\ref{ineq:pruneTheory2})
lies inside of $B(q,3\gamma)$, and the NN of $q$ will appear on one of
the lists $L_r$.  If the algorithm examined every $x\in L_r$, it would
be guaranteed to find the nearest neighbor, but we cannot say how many
total points it will examine.  However, the algorithm does not need to
examine the entire list, as we now show.

Suppose that $x$ is $q$'s NN; what is
the maximum distance it can be from its representative $r$?
From the triangle inequality, $\rho(x,r) \leq \rho(x,q) + \rho(q,r)$.
But since $x$ is $q$'s NN, and since $r\in X$, $\rho(x,q) \leq
\gamma$.  The other term is bounded by $3\gamma$ on account of
(\ref{ineq:pruneTheory2}).  Hence 
the nearest neighbor $x$ must lie within $4\gamma$ of its
representative. 

Since any considered representative $r$ satisfies $\rho(q,r)\leq
3\gamma$ (by (\ref{ineq:pruneTheory2})) and the NN $x$ of $q$ is
within $4 \gamma$ of $r$, the triangle inequality implies that
$\rho(q,x) \leq 7\gamma$.  
\end{proof}

We have shown that the search algorithm only needs to compute
distances from $q$ to points $x$ which are within distance $4\gamma$
of their representative.  Hence, if the lists $L_r$ are stored in
sorted order according to the distance to $r$, the search algorithm
can simply ignore all points $x$ more than distance $4\gamma$ from
their representative.\footnote{For the purpose of scheduling on some
  systems, it may be advantageous to compute how much of each list
  $L_r$ must be explored before the second brute force operation
  begins.  This can be done in time $O(\log{n})$ per list.}

Finally, we bound the expected number of examined points.  From Claim
\ref{cl:ballBnd}, all examined points lie in $B(q,7\gamma)$.
Applying the expansion rate condition, we have that
\begin{align}
|B(q,7\gamma)| \leq |B(q,8\gamma)| \leq c^3|B(q,\gamma)|. \label{ineq:cardBnd}
\end{align}
From Claim \ref{cl:numPts}, $\E|B(q,\gamma)| = n_r/n$, which we can
plug into (\ref{ineq:cardBnd}).  
As each $x$ only appears on one list $L_r$, each $x$ is only compared
to $q$ once, implying that (\ref{ineq:cardBnd}) bounds not only the
\emph{cardinality} of the examined set points, but also the number of
computations (in the second step).

Putting everything together, we have the following theorem.
\begin{theorem}
The expected number of points examined in the second stage of the
exact search algorithm is $c^3 n/n_r$, which is
$O(c^{3/2}\sqrt{n})$ for the standard parameter setting. \label{thm:exact}
\end{theorem}
Since the time for the first brute force step was also $O(c^{3/2}\sqrt{n})$,
we have shown that the expected runtime of the exact search 
algorithm is $O(c^{3/2}\sqrt{n})$.   

\subsection{One-Shot Search}
The one-shot search algorithm is considerably simpler than the exact
search algorithm, and also uses a slightly different data structure
configuration.  In particular, the algorithm searches only a single
representative list per query, and the ownership lists of the RBC will
usually overlap.  Unlike the exact search algorithm, the one-shot
algorithm only returns the NN with high probability, similar to
locality sensitive hashing \cite{indyk98}.  

With these differences in mind, the resulting time complexity bound is
actually quite similar to the bound in Theorem \ref{thm:exact}.
Recall that there are two parameters governing its run time: $n_r$,
the number of representatives; and $s$, the number of points assigned
to each representative.  Hence the time complexity of the one-shot
search algorithm is $O(n_r + s)$.  The following theorem describes the
setting of these parameters to guarantee a high probability of
success.  
\begin{theorem}
Set the parameters
\begin{align*}
n_r\;=\;s\;=\;c\sqrt{n}\cdot\sqrt{\ln\frac{1}{\delta}}.
\end{align*}
Then the one-shot algorithm returns the correct NN with probability at
least $1-\delta$. \label{thm:one-shot}
\end{theorem}
The proof is in Appendix \ref{ap:one-shot}.

Hence the time complexity of the one-shot algorithm is $O(c\sqrt{n})$
times a factor dependent on the desired success rate.  Notice that the
dependence on $c$ is lower for the one-shot algorithm than the exact
algorithm; this reduced complexity seems to be reflected in the
experiments.  


\section{Experiments}
We perform several sets of experiments to demonstrate the
effectiveness of our methods.  The first, and probably most
important, set of experiments demonstrate the performance benefit of
the RBC on a 48-core machine, as compared to a brute force
implementation (\S\ref{exp:main}).  These experiments show that the RBC
significantly reduces the work required for NN search and that it
parallelizes effectively.  

The second set of experiments demonstrates that the RBC is effective
on graphics hardware (\S\ref{exp:gpu}).  It is challenging to deploy
data structures on such hardware, but very important because of the
ubiquity of GPUs in scientific and database systems.  

The final set of experiments compares the performance of the RBC to
the Cover Tree on a desktop machine (\S\ref{exp:covtree}).  These
experiments demonstrate that the exact search algorithm is competitive
with the state-of-the-art even on a machine with a low degree of
parallelism.

\subsection{Experimental setup}
Our CPU implementation of the RBC was written in C and parallelized
with OpenMP.  Our GPU implementation was written in C and CUDA.  
Both are available for download from the author's website.

In very low-dimensional spaces, basic data structures like kd-trees
are extremely effective, hence the challenging cases are data that is
somewhat higher dimensional.  We experiment on several different data
sets over a range of dimensionalities.  Table \ref{tbl:data} provides
a quick overview; we describe a few details next.
\begin{table}
\centering
\begin{tabular}{lll}
Name & Num pts & Dim \\\hline
Bio & 200k & 74\\
Covertype & 500k & 54\\
Physics & 100k & 78 \\
Robot & 2M & 21 \\
TinyIm & 10M & 4-32 
\end{tabular} \caption{Overview of data sets.} \label{tbl:data}
\end{table}

The Bio, Covertype, and Physics data sets are standard benchmark data
sets used in machine learning and are available from the UCI
repository \cite{fa.10}.  They have been used to benchmark 
NN search previously \cite{BKL06,cd.07,rlog.09}.  The Robot data was generated
from a Barret WAM robotic arm; see \cite{np.10}.
The TinyIm data set is taken from the Tiny Images database, which
is used for computer vision research \cite{trf.2008}.  We took the
image descriptors and reduced the dimensionality using the method of
random projections.\footnote{This dimensionality reduction technique
  approximately preserves the lengths of vectors, and hence is a
  useful preprocessor for NN search; see e.g. \cite{LMGY04}.
  The technique is formally justified by the
  Johnson-Lindenstrauss Lemma \cite{jl.84}.}  We 
experimented with dimensionalities of 4, 8, 16, 32.   For all
experiments, we measured distance with the $\ell_2$-norm
(i.e. standard Euclidean distance), which is appropriate for this
data.  

We perform the first set of CPU experiments on a 48-core/4-chip AMD
server machine.  Each chip is a 12-core AMD 6176SE processor, and is
divided into two 6-core segments.  This machine has a high core count,
so it is a good system to test the scalability of our algorithms.  

We compare to the Cover Tree on a quad-core Intel Core i5 machine,
which is a reasonable representative for a mid-range desktop.

Our GPU experiments are run on a NVIDIA Tesla c2050 graphics card.
This card is well-suited to general and scientific computation as it
has somewhat more memory than standard graphics cards (3GB), and
provides general-purpose architectural features such as error-correcting memory and
caching.\footnote{On the other hand, in limited experiments with
  (much cheaper) consumer-grade GPUs, we noticed that the
  runtimes were not dramatically different.}  We previously described 
the details of our GPU implementation in a workshop paper \cite{c.10}.

\subsection{48-core experiments}\label{exp:main}
\begin{figure*}
\centering
\begin{tabular}{ccc}
\includegraphics[width=0.3\linewidth,height=!]{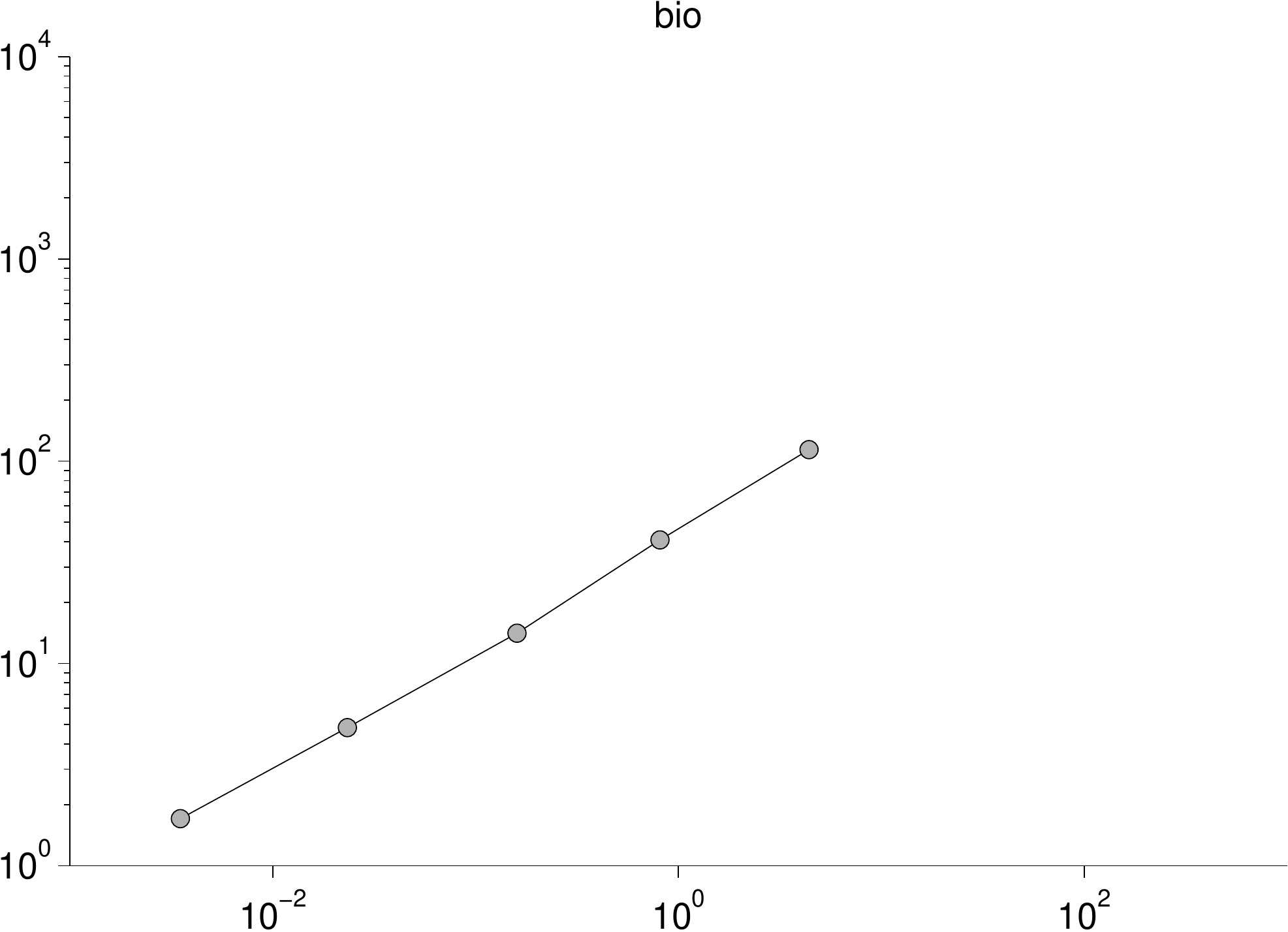}
& 
\includegraphics[width=0.3\linewidth,height=!]{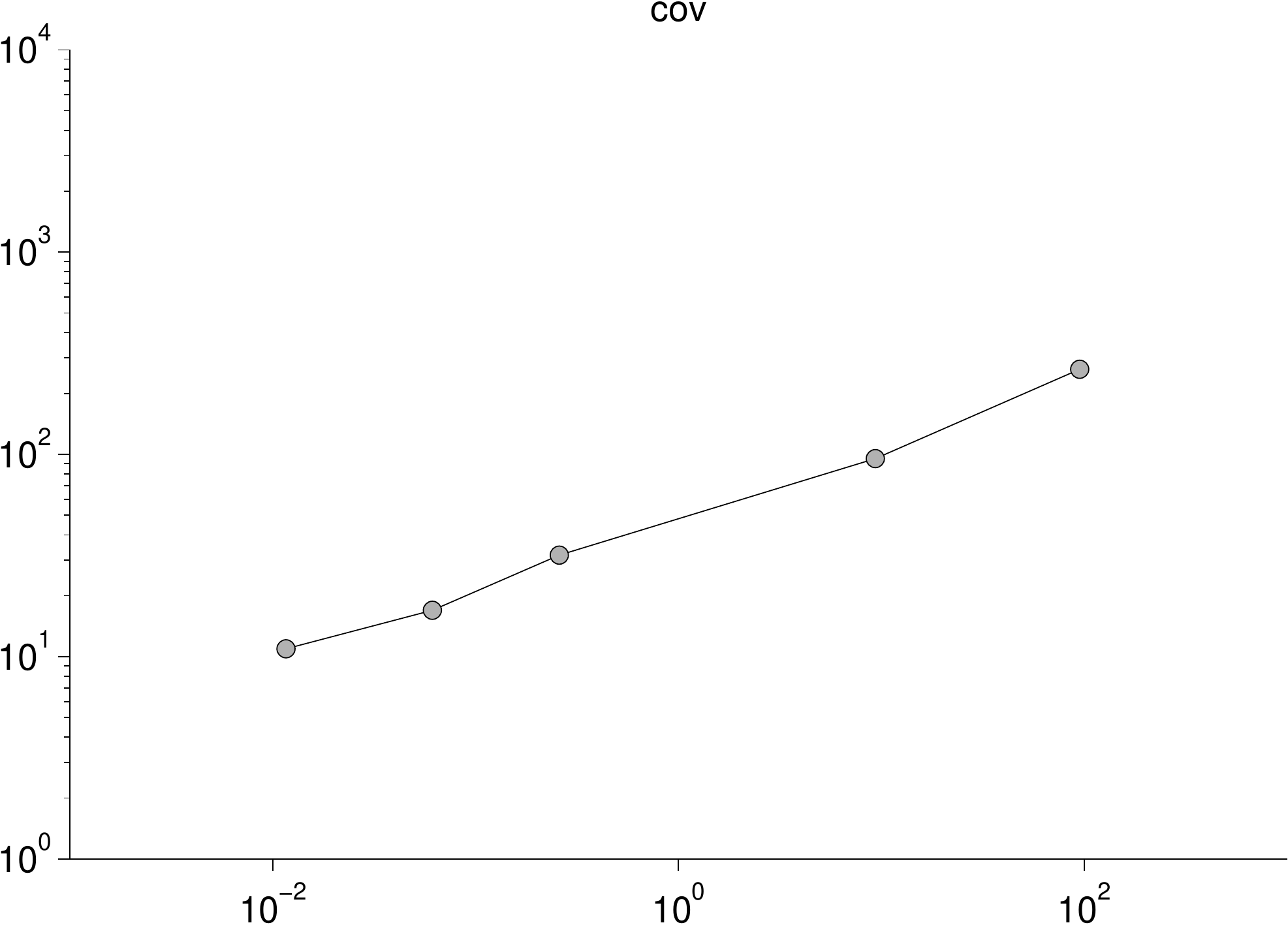}
& 
\includegraphics[width=0.3\linewidth,height=!]{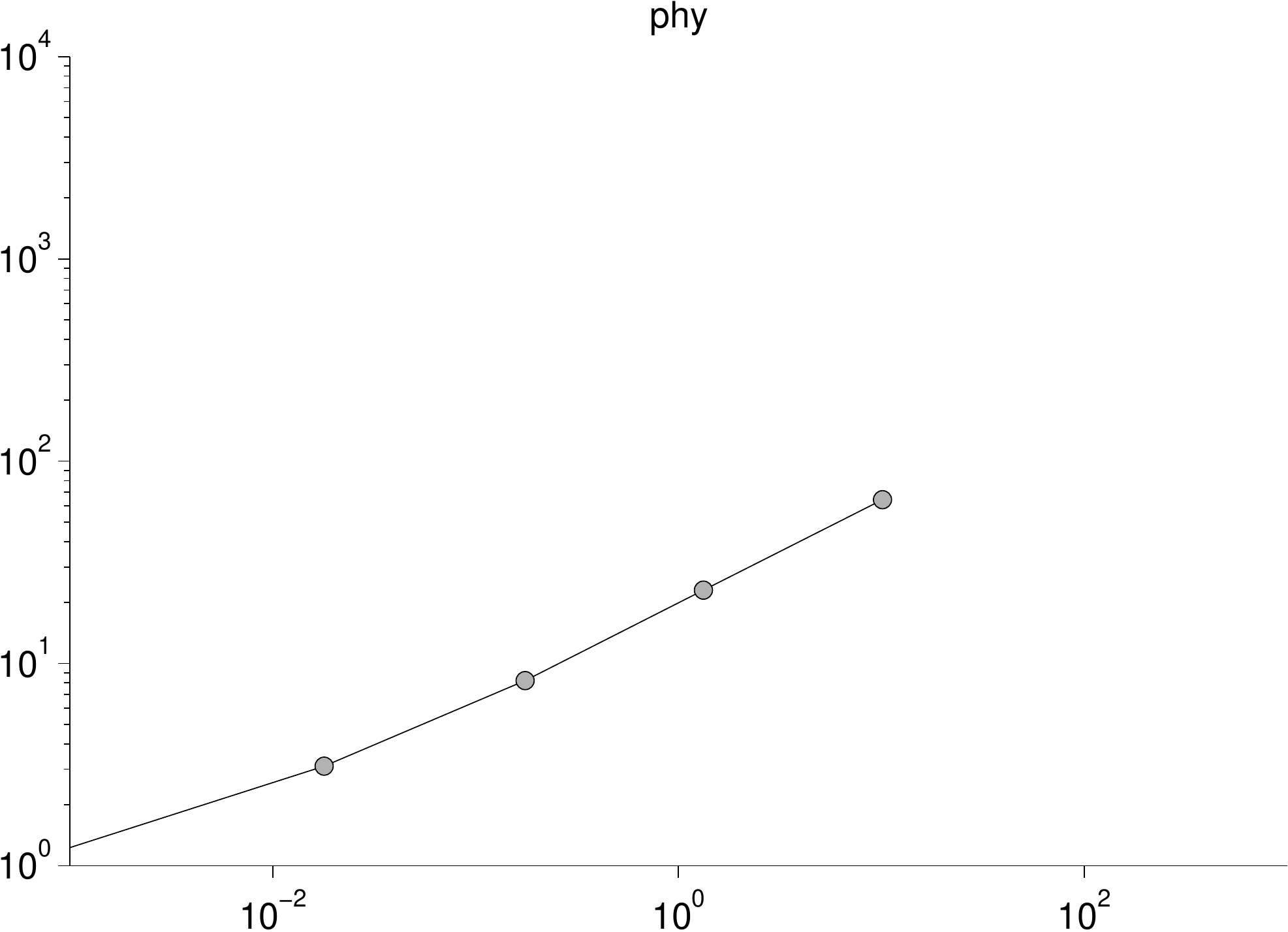}
\\
\includegraphics[width=0.3\linewidth,height=!]{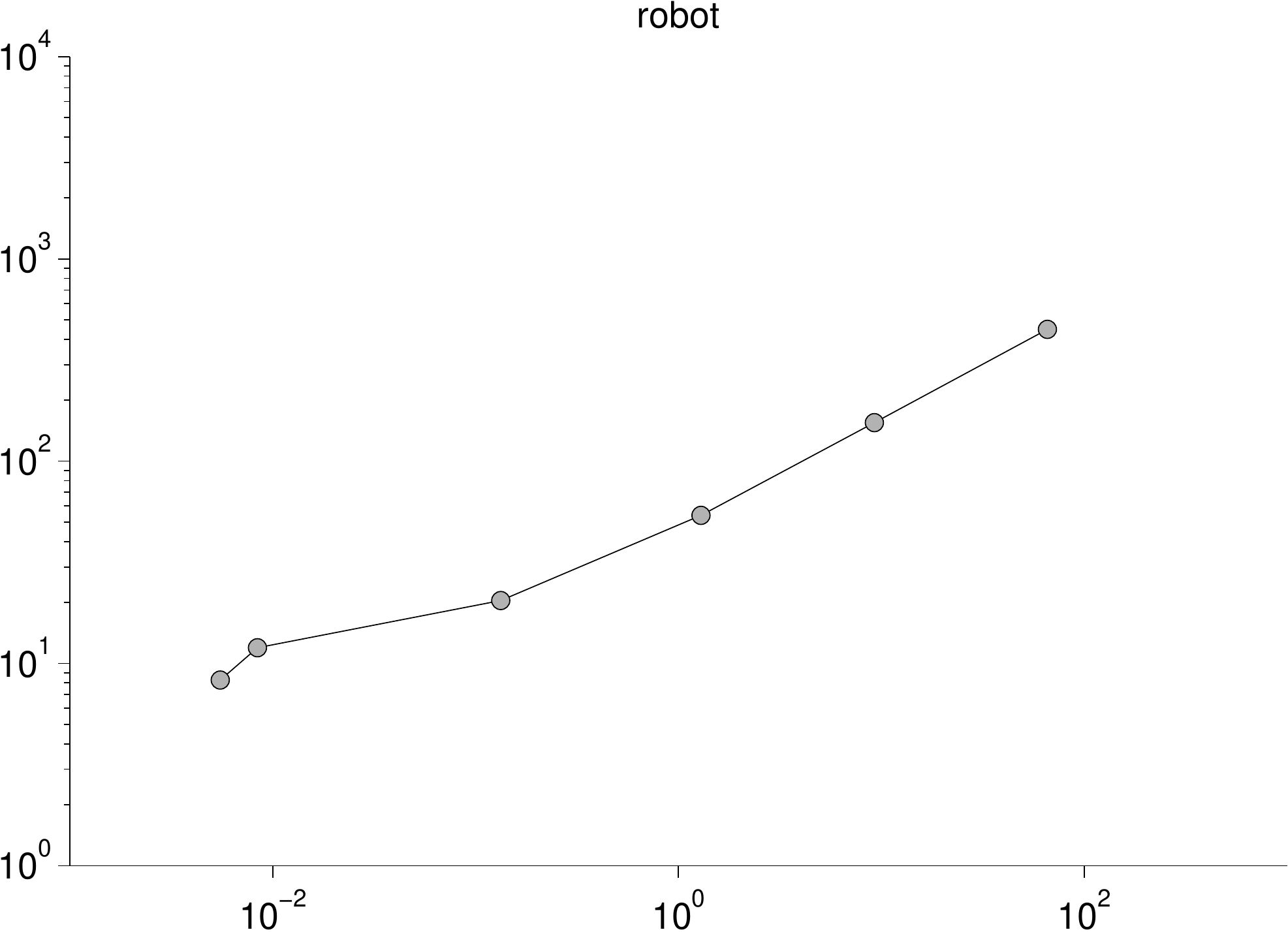}
& 
\includegraphics[width=0.3\linewidth,height=!]{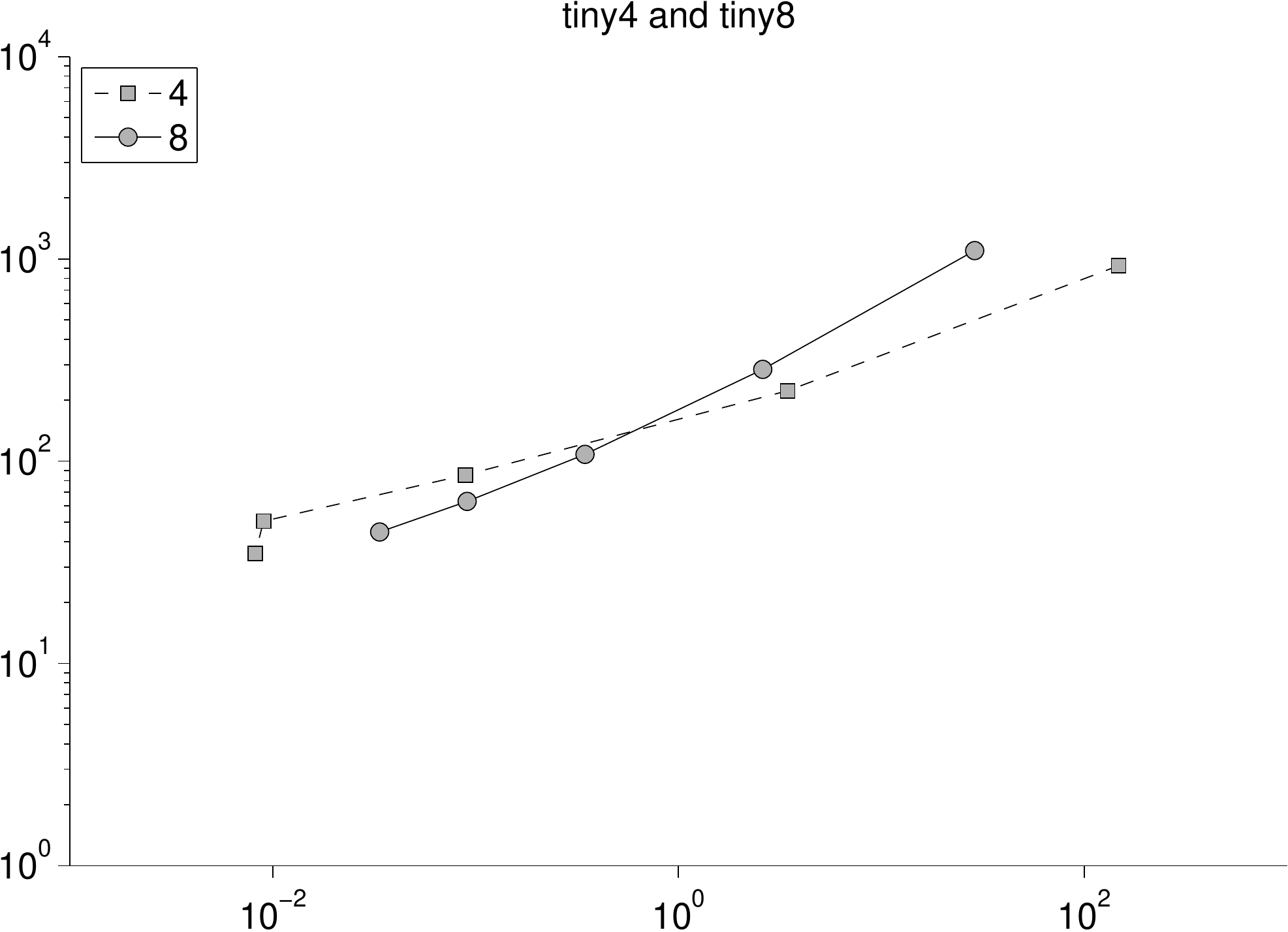}
& 
\includegraphics[width=0.3\linewidth,height=!]{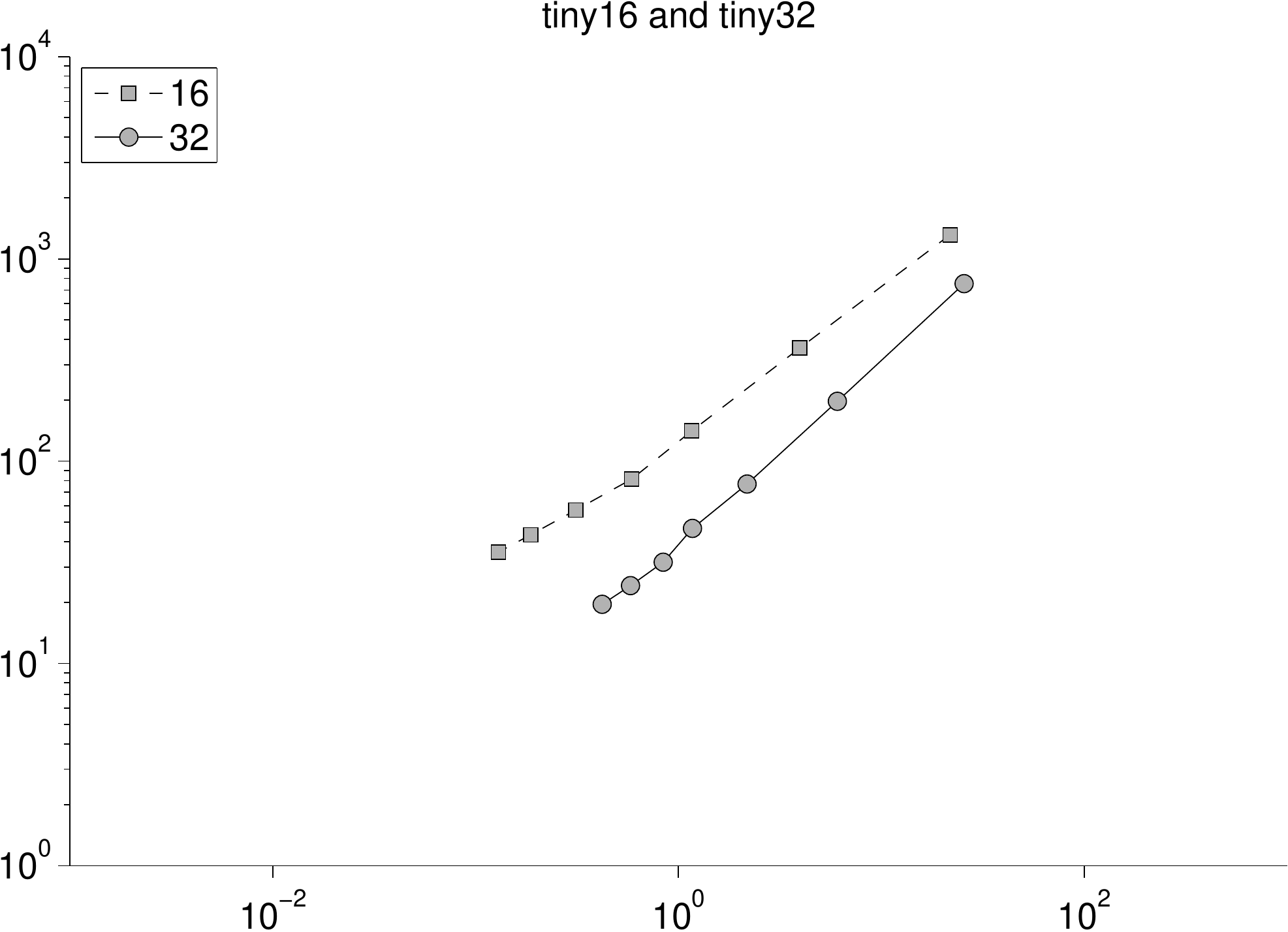}
\end{tabular} \caption{Results of the one-shot algorithm.  This is a
  log-log plot of the speedup as a function of the error rate.   The
  $x$-axis is logarithmic and runs from $10^{-3}$ to $10^3$ and
  signifies the average (over queries) rank of the returned result.
  For example, a rank of $10^0$ indicates that on average the
  algorithm returns the 2nd NN.  
  The $y$-axis is also logarithmic and runs from $10^0$ (no speedup)
  to $10^4$ (10000x speedup).}\label{fig:defeat}
\end{figure*}

We compare the performance of our methods to brute force search on the
48-core machine.  As far as we know, there is no readily available
accelerated NN method for such a machine.  Furthermore, brute force is
already quite fast because of the raw computational power.  

First, we look at the performance of the exact search algorithm,
which is guaranteed to return the exact NN.  Figure \ref{fig:exactBar}
shows the results.  We are getting a strong 
speedup of up to two orders of magnitude, despite the challenging
hardware setting and the reasonably high dimensionality of the
data.\footnote{There is one parameter in this algorithm; namely, the
  number of representatives chosen.  The retrieval times were not
  particularly sensitive to this choice; see Appendix \ref{ap:exact} for a
  detailed examination of the parameter setting.}

\begin{figure}
\centering
\includegraphics[width=.8\linewidth, height=!]{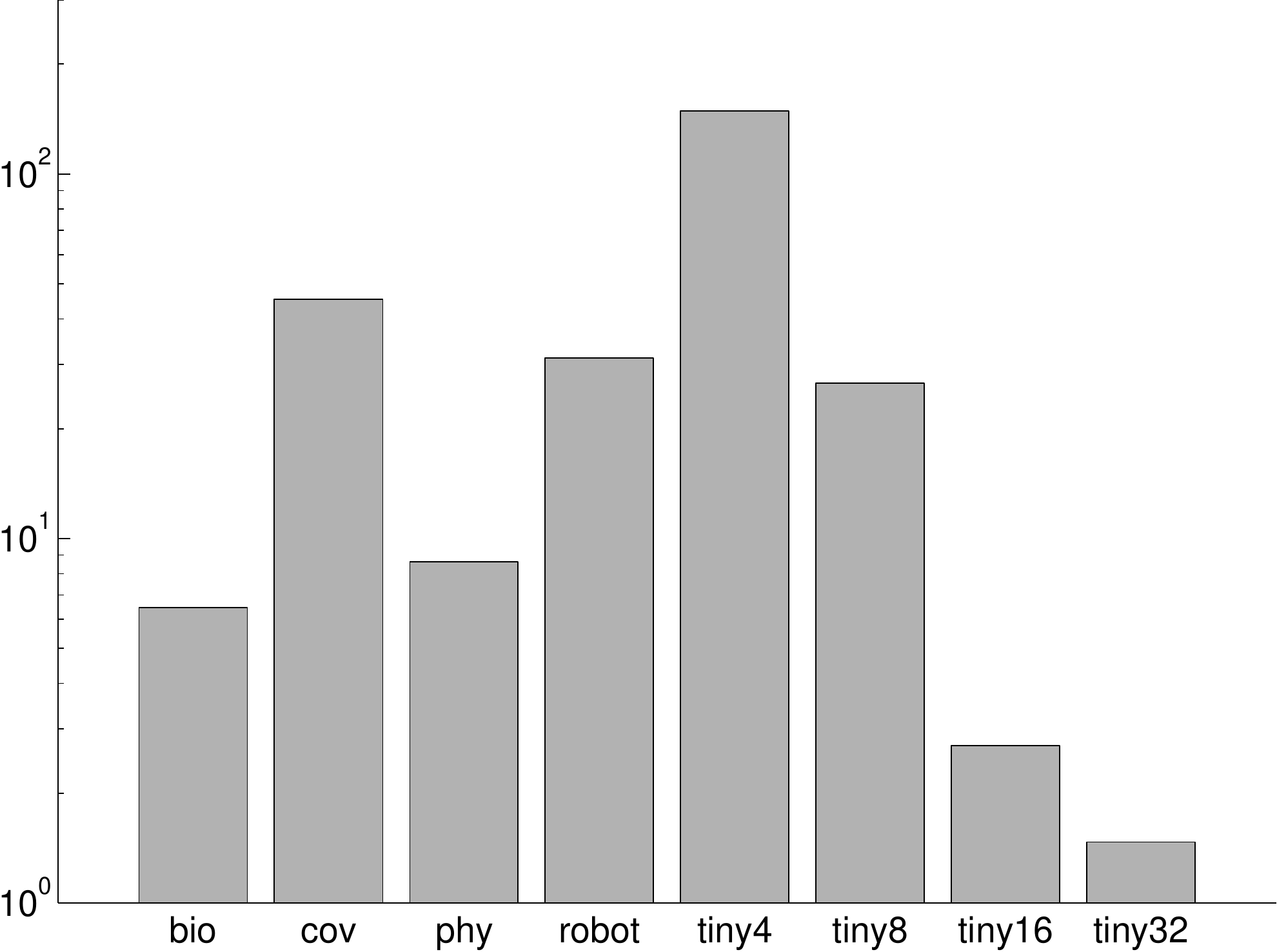} 
\caption{Speedup of exact search over brute force.} \label{fig:exactBar}
\end{figure}

Next we consider the one-shot search algorithm.  As developed in the
theory section, we set $n_r$ (the number of representatives) and $s$
(the number of points owned per representative) equal to one another.
The parameter allows one to trade-off between the quality of the
solution and time required; we scan over this parameter to show the
trade-off.  This algorithm is not guaranteed to return a nearest neighbor, so we
must evaluate the quality of the solution.  A standard error measure is
the \emph{rank} of the returned point: i.e., the number of database
points closer to the query than the returned point \cite{rlog.09}.  A
rank of 0 denotes the exact NN, and rank of 1 denotes the second NN,
and so on. 

Figure \ref{fig:defeat} shows the results.  The speedup achieved in
these experiments is quite significant;  even with a rank 
around $10^{-1}$ (very close to exact search), the \emph{worst} speedup is an
order of magnitude.  In applications  where a small amount of error is
tolerable, the one-shot search algorithm can provide a massive
speedup, even better than our exact search algorithm in some cases.
In many applications, e.g. in data mining, there is considerable
uncertainty associated with the data, so a small amount of error in
the NN retrieval is not important.    
\subsection{GPU experiments} \label{exp:gpu}
GPUs have impressive brute force search performance \cite{gdb.08}. 
However, the GPU architecture makes efficient data structure design
quite difficult.  In particular, GPUs are vector-style processors with
limited branching ability; hence conditional computation typically
seriously under-utilizes these devices.  

We show that our RBC one-shot algorithm provides a substantial speedup
over the already-fast brute force search on a GPU.  Table
\ref{fig:gpuResults} shows the results.  We show only the speedups, as
the error rate is the same as that of the CPU experiments.  The
parameter was set to achieve an error rate of roughly $10^{-1}$ (refer
back to Figure \ref{fig:defeat}).  Our method is clearly very
effective in this setting; despite the challenging hardware design, it
provides a one-to-two order of magnitude speedup on all datasets.


\begin{table}
\centering
\begin{tabular}{ll}
Data & Speedup \\ \hline
Bio & 38.1  \\
Covertype & 94.6  \\
Physics & 19.0 \\
Robot & 53.2 \\
TinyIm4 & 188.4 \\
\end{tabular} \caption{GPU results: speedup of the one-shot algorithm
  over brute force search (both on the GPU). }\label{fig:gpuResults}
\end{table}

\subsection{Cover Tree comparison} \label{exp:covtree}
Finally, we investigate the performance of the RBC on a quad-core
desktop machine.  We compare to the Cover Tree, which has
state-of-the-art empirical performance on a sequential machine, and
which was developed under the same notion of intrinsic dimensionality
as the RBC.   

This is a challenging comparison for the RBC, as the system does not
require the extreme constraints under which the RBC was designed.  
In particular, this system can branch effectively and has a low core
count.  Hence such a 
system can handle an algorithm with a much more complex computational
structure, like the Cover Tree search algorithm.  Suprisingly, 
the performance of the RBC is at or near state-of-the-art even in this
setting.   

The available implementation of the Cover Tree is single-core
\cite{BKL06}, and a naive parallelization would not significantly
benefit it.  In particular, one could 
split the database into $p$ chunks (one for each core), run $p$ Cover
Tree searches in parallel, and then perform a $p$-way reduce.
However, since the dependence of the Cover Tree on the number of data
points is only $O(\log{n})$, doing so would only decrease the runtime
from $O(c^6\log{n})$ to, at best, $O(c^6 \log{\frac{n}{p}})$, which is
a very minor improvement.  Hence we run the Cover Tree only on one
core, but allow the RBC to use the whole machine.  

Table \ref{tbl:cov} shows the results.\footnote{We were unable to get the 
Cover Tree software to run on the full TinyIm data sets, so we reduced
the database size to 1M.}  The RBC is competitive on all of the datasets, and
significantly outperforms the Cover Tree on the three largest datasets.
Again, these results are surprising, as our exact search algorithm
is much simpler than the Cover Tree search algorithm, and since our methods
have the additional (significant) constraint that they must work on
highly parallel systems.  

We note that the RBC has a significantly lower theoretical dependence
on the dimensionality than the Cover Tree ($O(c^{3/2})$ vs $O(c^6)$).
This is reflected in the experiments; the two datasets that the
Cover Tree significantly outperforms the RBC on are very
low-dimensional: the Tiny4 data set is four-dimensional, and the
Covertype dataset has low intrinsic dimensionality \cite{BKL06}.  This
reduced dependence on dimensionality appears to be another advantage
of the RBC and would be interesting to explore in future work.  

\begin{table}
\centering
\begin{tabular}{lll}
Data & Cover Tree& RBC  \\ \hline
Bio  & 18.9  & {6.4}\\
Covertype  & {0.4} & 1.1 \\
Physics  & 1.9 & {1.7} \\
Robot & {4.6}& 5.1  \\
Tiny4 & {0.5}& 1.2 \\
Tiny8 & 14.6 & {3.3}\\
Tiny16 & 178.9 & {25.1} \\
Tiny32 & 387.0 & {67.9}
\end{tabular} \caption{Comparison of the Cover Tree and the exact RBC
  algorithm on a quad-core desktop machine.  Times shown are the total
  query time in seconds for 10k queries.} \label{tbl:cov} 
\end{table}

\section{Conclusion}
In this paper, we introduced techniques for metric similarity search
on parallel systems.  In particular, we demonstrated that the RBC
search algorithms significantly reduce the work required for 
NN retrieval, while being structured in such a way that can be easily
implemented on parallel systems.  Our experiments show that these
techniques are practical on a range of modern hardware.  The theory
behind the RBC shows that the data structure is broadly effective.

Our code is available for download.  These implementations supply
additional low-level details on implementing the RBC.  Moreover, they
are practical tools for many NN search problems.  

An interesting direction for future work  is to explore the
performance of the RBC in a distributed or multi-GPU environment.  The
RBC data structure suggests a simple distribution of the database
according to the representatives that could be quite effective in such
environments.  There are many interesting details for study here, such
as I/O and communication costs, and the connection to distributed 
programming paradigms.  Furthermore, a distributed implementation
would be broadly useful in practice.  


\bibliographystyle{abbrv}
\bibliography{biblio}  

\appendix

\begin{figure*}
\centering
\begin{tabular}{ccc}
\includegraphics[width=0.3\linewidth,height=!]{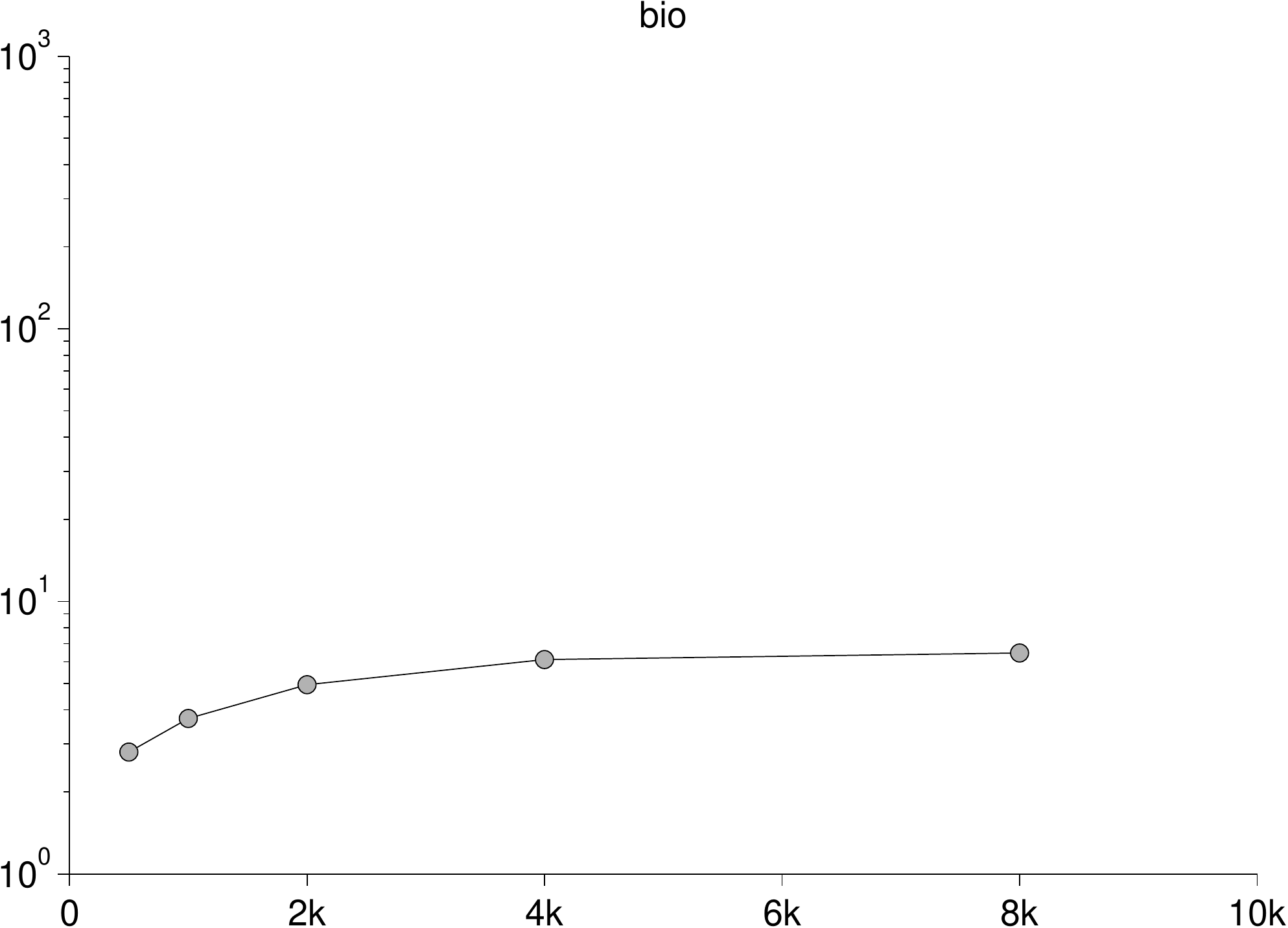}
& 
\includegraphics[width=0.3\linewidth,height=!]{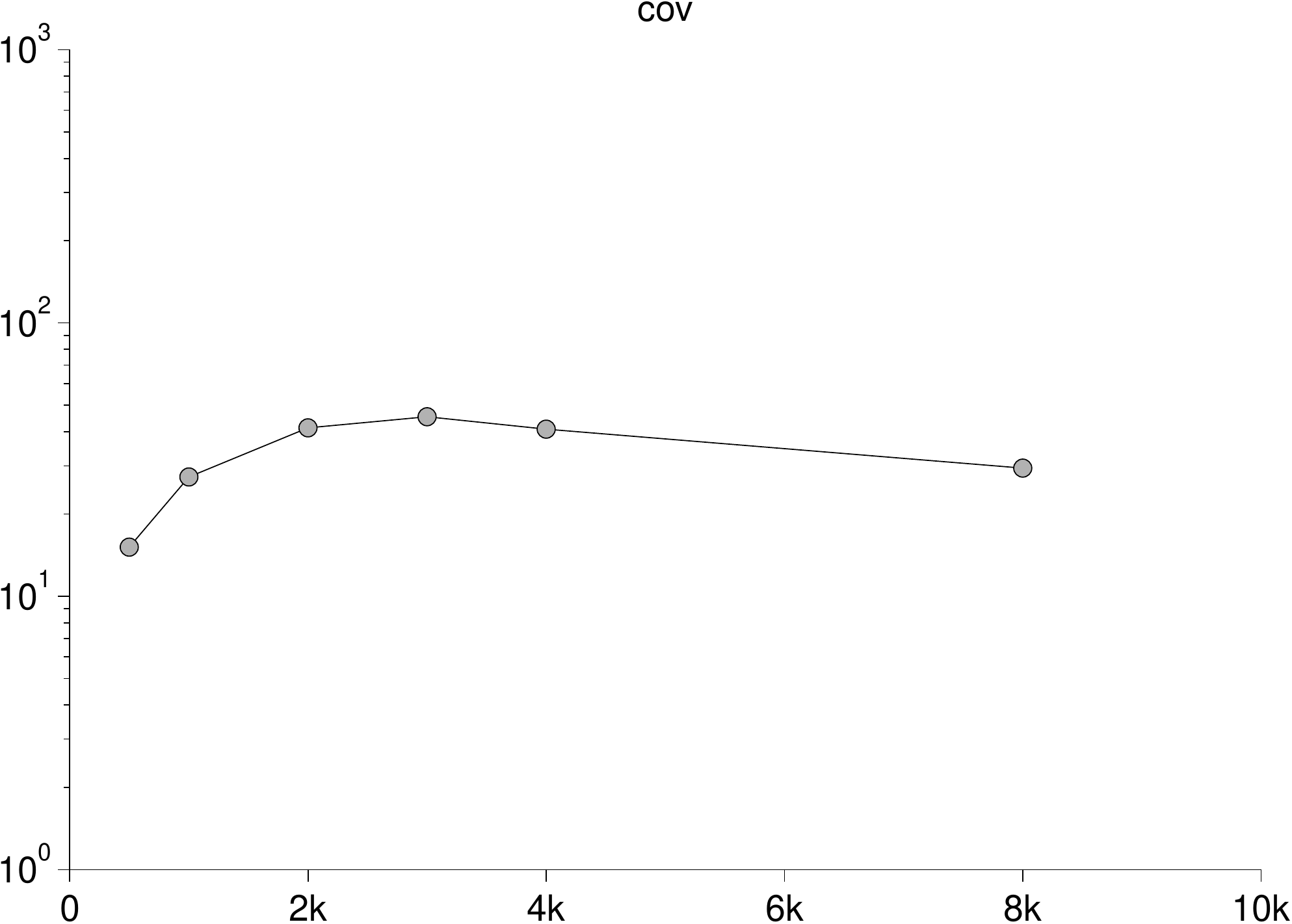}
& 
\includegraphics[width=0.3\linewidth,height=!]{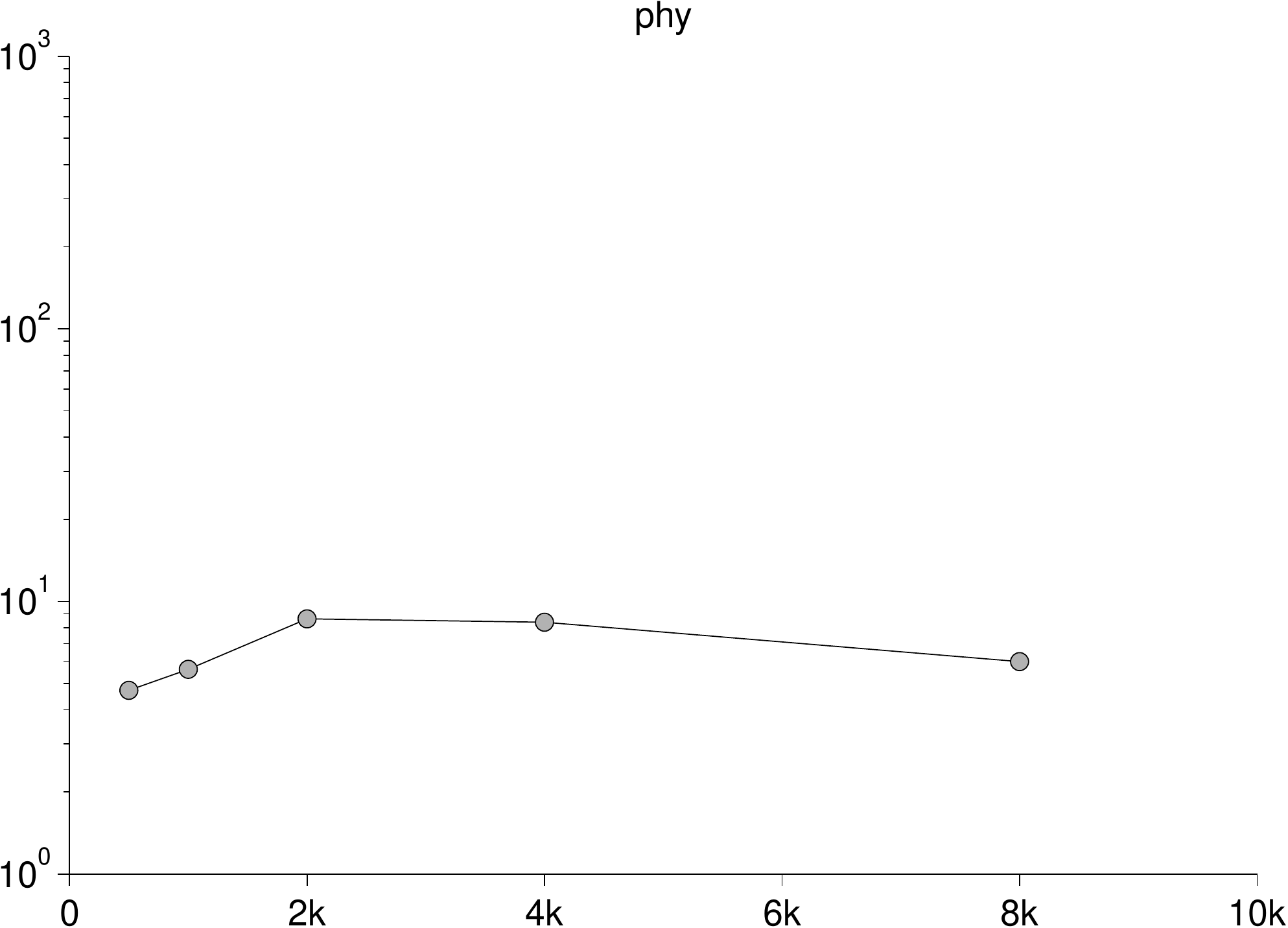}
\\
\includegraphics[width=0.3\linewidth,height=!]{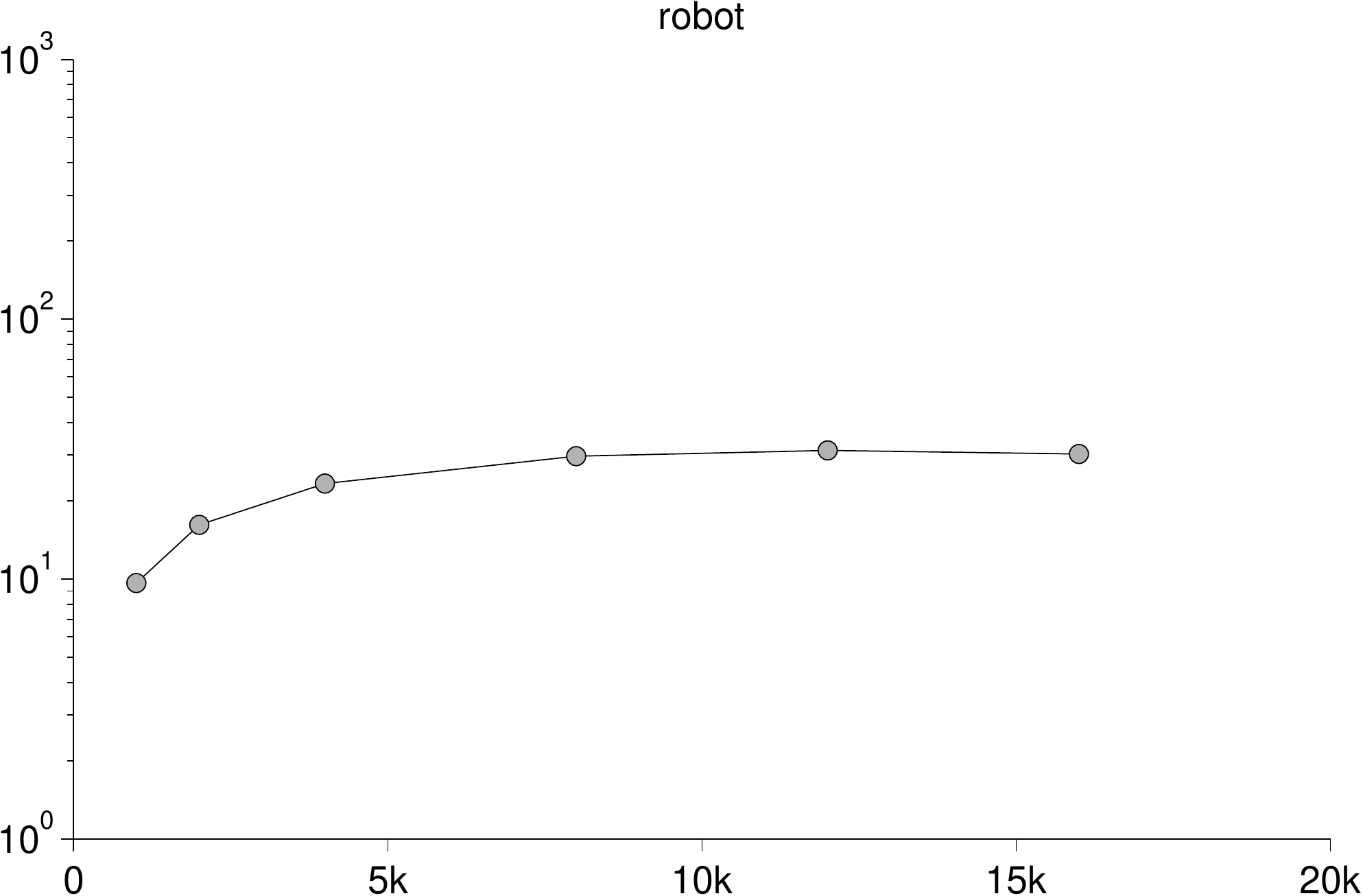}
& 
\includegraphics[width=0.3\linewidth,height=!]{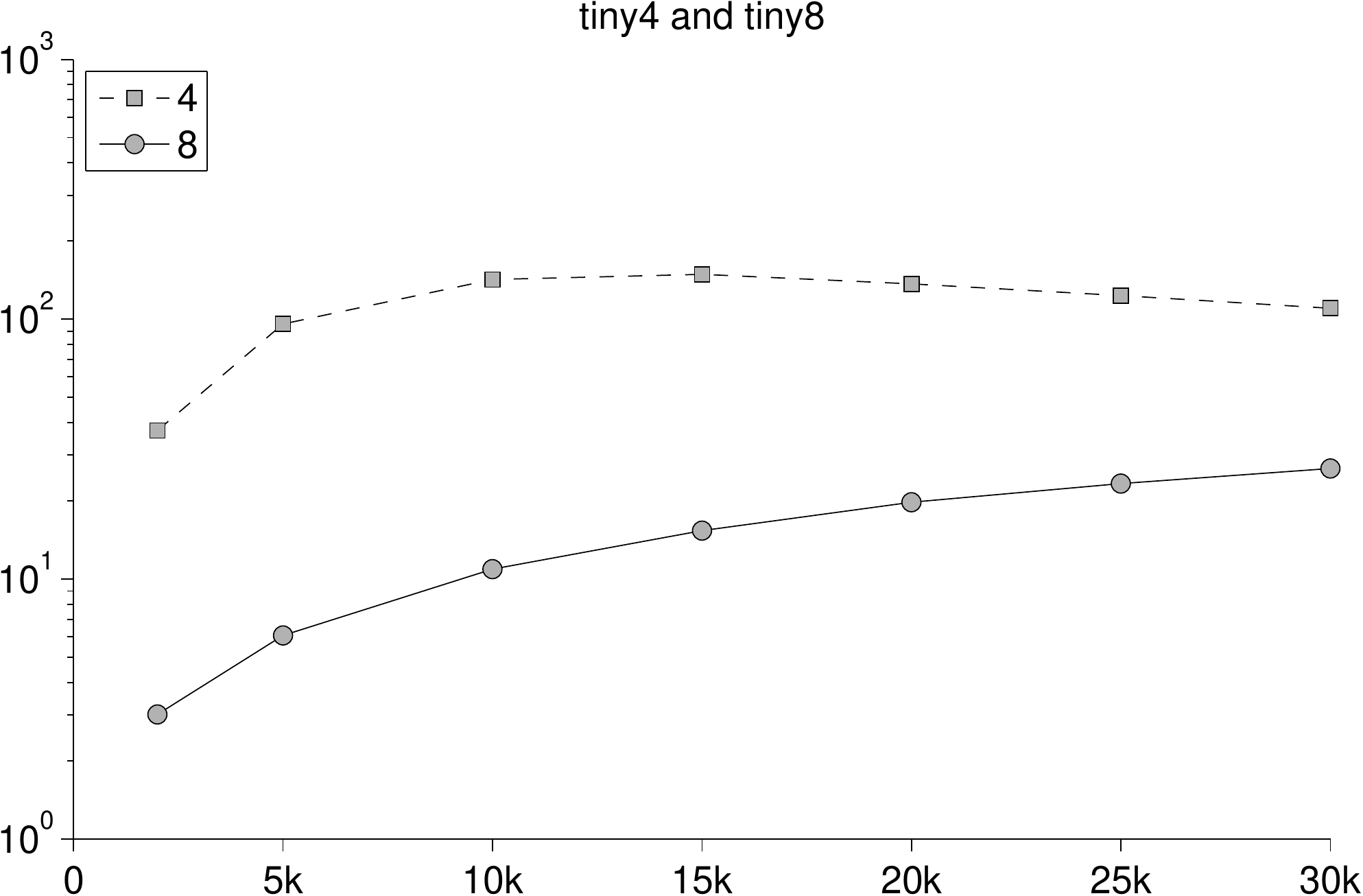}
& 
\includegraphics[width=0.3\linewidth,height=!]{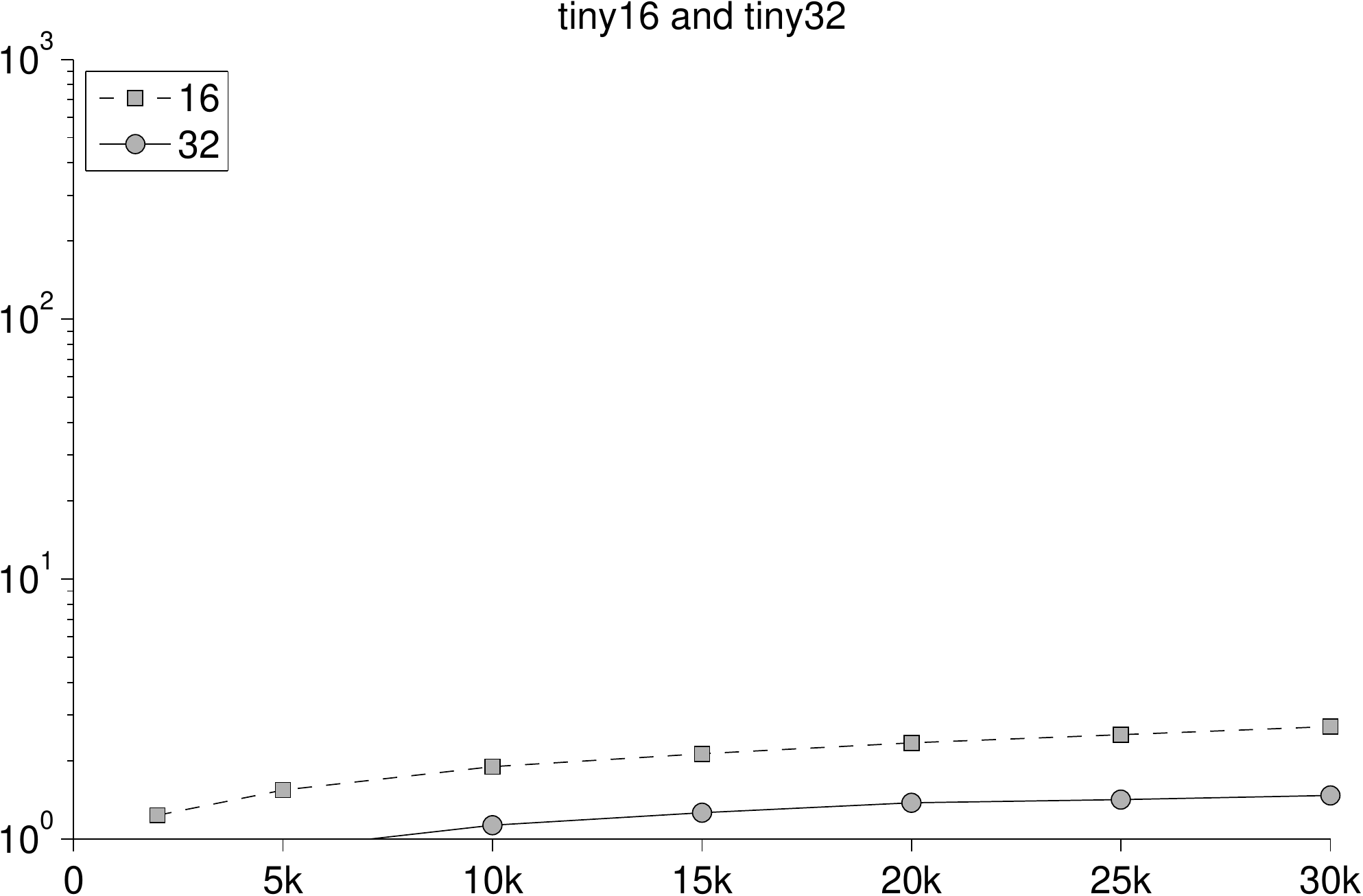} 
\end{tabular} \caption{Results for exact search.  The $y$-axis is a
  logarithmic scale that varies from $10^0$ (no speedup) to $10^3$
  (1000x speedup).  The $x$-axis is the number of representatives
  chosen.} \label{fig:exact} 
\end{figure*}
\section{Proof of Lemma}\label{ap:lem-proof}
We restate Lemma \ref{lem:prune2} and prove it.
\begin{lemma}
Let $R \subset X$ and assign each $x\in X$ to its nearest $r\in R$.
Let $\gamma = \min_{r\in R} \rho(q,r)$ (i.e., $\gamma$ is the distance
to $q$'s NN in $R$).  Then, if some $r^*\in R$ owns the nearest neighbor
to $q$ in $X$, it must satisfy
\begin{align*}
\rho(q,r^*) \leq 3\gamma.
\end{align*}
\end{lemma}
\begin{proof}
Suppose that $x$ is $q$'s NN in $X$---i.e. $\rho(q,x) \leq \rho(q,y)$
for all $y\in X$---and that $r^*$ owns $x$ ($r^*$ is $x$'s NN among
$R$).  Furthermore, let $r$ be $q$'s NN in $R$.  Since $R\subset X$ and
$\rho(q,r) = \gamma$, $\gamma$ gives an upper bound on the distance to
$q$'s NN; hence $\rho(q,x) \leq \gamma$.  Using this bound along
with the triangle inequality gives $\rho(x,r) \leq 2\gamma$, but since
$\rho(x,r^*) \leq \rho(x,r)$, we have $\rho(x,r^*) \leq 2\gamma$ as
well.  Applying the triangle inequality to the bounds on $\rho(x,r^*)$
and $\rho(q,x)$ yields the lemma.  
\end{proof}

\section{Proof of the one-shot theorem}\label{ap:one-shot}
We will use the following lemma, which is well-known. We include a
proof for completeness.  
\begin{lemma}  Suppose that $\rho(q,r) = \gamma$.  Then 
\begin{align*}
  B(q,\gamma) \subset B(r,2\gamma) \subset B(q,4\gamma).
\end{align*} \label{lem:sandwich}
\end{lemma}
\begin{proof}
Let $x\in  B(q,\gamma)$.  Then $\rho(x,r) \leq \rho(x,q) + \rho(q,r)
\leq 2\gamma$, proving the first inequality.  For the second, let $x
\in B(r,2\gamma)$.  Then $\rho(x,q) \leq \rho(x,r) + \rho(r,q) \leq
2\gamma + \gamma \leq 4\gamma$.
\end{proof}

Now we restate and prove Theorem \ref{thm:one-shot}.
\begin{theorem}
Set the parameters
\begin{align*}
n_r\;=\;s\;=\;c\sqrt{n}\cdot\sqrt{\ln\frac{1}{\delta}}.
\end{align*}
Then the one-shot algorithm returns the correct NN with probability at
least $1-\delta$.  
\end{theorem}
\begin{proof}
If a query $q$ lies within distance $\psi_r/2$ of a representative
$r$, then its nearest neighbor $x_q$ is guaranteed to be in $L_r$.
This follows simply from the triangle inequality:
\begin{align*}
\rho(r,x_q) &\leq \rho(r,q) + \rho(q,x_q)\\
&\leq \rho(r,q) + \rho(q,r) \quad\mbox{(since $x_q$ is $q$'s NN)}\\
&\leq \psi_r.
\end{align*}
Thus the algorithm fails only if $\rho(q,r) > \psi_r/2$ for all $r$
(and in particular that nearest one).  We bound the probability of
this occurring.

Let $\gamma = \rho(q,r)$, where $r$ is $q$'s NN.  Again, assume that
$\gamma > \psi_r/2$.  We have that $B(r,\psi_r) \subset B(r,
2\gamma)$, and $B(r,2\gamma) \subset B(q,4\gamma)$ by Lemma
\ref{lem:sandwich}.  Hence we can apply the expansion condition to get
that 
\begin{align*}
|B(r,\psi_r)| \leq |B(q,4\gamma)| \leq c^2 |B(q,\gamma)|.
\end{align*}
In particular, $|B(q,\gamma)| \geq 1/c^2 |B(r,\psi_r)|$; of course
$B(r,\psi_r)$ contains the points of $L_r$, which we chose to be
$s$.  Thus there are at least $s/c^2$ points closer to $q$ than $r$.
What is the probability that none of them were chosen as
representatives?

Each point is chosen with probability $n_r/n$, so the probability that
none are chosen is 
\begin{align*}
\left( 1- \frac{n_r}{n}\right)^{s/c^2} \leq e^{-\frac{n_rs}{c^2n}},
\end{align*}
which follows from the inequality $(1-t/r)^r \leq e^{-t}$.  If we plug
in $n_r = s = \sqrt{\eta c^2 n }$, we get a probability of failure
bounded by 
\begin{align*}
e^{-\frac{n_rs}{c^2n}} = e^{-\eta}.
\end{align*}
The theorem follows by setting $\delta = e^{-\eta}$ and rearranging.  
\end{proof}

\section{Exact search experiments}\label{ap:exact}
There is a single parameter to set for the exact search algorithm (the
number of representatives).  Here we report the results over a fairly
wide range of parameters to support the experimental results.  Note
that the search time is relatively stable to this setting.  See Figure
\ref{fig:exact} for the results.

\end{document}